\documentclass[twoside]{article}

\usepackage[preprint]{aistats2026}

\usepackage{soul}
\usepackage[utf8]{inputenc} 
\usepackage{url}            
\usepackage{booktabs}       
\usepackage{amsfonts}       
\usepackage{microtype}      
\usepackage{xcolor}         
\usepackage{comment}
\usepackage{subcaption}
\usepackage{natbib}
\usepackage{enumitem}
\usepackage{balance}
\usepackage{listings}

\lstdefinelanguage{diff}{
    morecomment=[f][\color{red}]{-},
    morecomment=[f][\color{green}]{+},
    morecomment=[f][\color{gray}]{!}
}

\lstdefinelanguage{markdown}{
  morekeywords={####,*,-,>},      
  sensitive=false,
  morecomment=[l]{\%},         
  morestring=[b]",             
}

\usepackage{amsmath}
\usepackage{amssymb}
\usepackage{mathtools}
\usepackage{amsthm}
\allowdisplaybreaks

\theoremstyle{plain}
\newtheorem{theorem}{Theorem}[section]

\newtheorem{lemma}[theorem]{Lemma}

\theoremstyle{definition}

\theoremstyle{remark}
\newtheorem{remark}[theorem]{Remark}

\begin{document}

\twocolumn[

\aistatstitle{Iterative Data Curation with Theoretical Guarantees}

\aistatsauthor{ Väinö Yrjänäinen \And Johan Jonasson \And Måns Magnusson }

\aistatsaddress{ Department of Statistics, \\ Uppsala University \And Department of Mathematical Sciences \\ 
Chalmers University of Technology
\And  Department of Statistics, \\ Uppsala University  } ]

\begin{abstract}
In recent years, more and more  large data sets have become available.
Data accuracy, the absence of verifiable errors in data, is crucial for these large materials to enable high-quality research, downstream applications, and model training.
This results in the problem of how to curate or improve data accuracy in such large and growing data, especially when the data is too large for manual curation to be feasible. 
This paper presents a unified procedure for iterative and continuous improvement of data sets. 
We provide theoretical guarantees that data accuracy tests speed up error reduction and, most importantly, that the proposed approach will, asymptotically, eliminate all errors in data with probability one. We corroborate the theoretical results with simulations and a real-world use case.
\end{abstract}

\section{Introduction}



High-quality data is crucial for making good predictions \citep{jain2020overview,budach2022effects}, valid statistical inference \citep{hurtado2022documents, bernardi2023data,miller2025guidelines}, decision making \citep{wang2024overview}, and training reliable predictive models \citep{gunasekar2023textbooks}.
Moreover, perhaps the most important aspect of data quality is data accuracy \citep[p.~27]{olson2003data}, which can be seen as the absence of erroneous values.
In this article, we focus on data accuracy and define errors as discrepancies between the stored values and some verifiable ground truth. These may be OCR errors, yielding a discrepancy between printed physical media and its digital representation; label errors, where an image, piece of text, or other piece of data is labeled as the incorrect category; a discrepancy between a value in a database describing a rental apartment and the real measurements of the apartment; missing or invalid values in a database table; or a number of other ways the stored pieces of data differ from some well-defined and verifiable ground truth.

Despite the importance of data accuracy, data sets commonly used in machine learning, such as Fashion MNIST \citep{xiao2017fashion}, Common Crawl \citep{commoncrawl}, and the Wikipedia corpus \citep{wiki2023index}, often provide few guarantees of their accuracy. Instead, even the most commonly used data sets contain a notable amount of errors \citep{northcutt2021pervasive}.
This is often the case in research data sets as well, eg. parliamentary data \citep{gentzkow2018congressional, erjavec2023parlamint}, legal research datasets \citep{ostling2023cambridge, butler-2025-open-australian-legal-corpus}, or literature corpora \citep{davies_coca} is well-known to contain errors or lack accuracy guarantees.

Improving and verifying the quality of large data sets is a difficult task \citetext{\citealp{kaddour2023challenges}}.
Limited resources force a balance between accuracy, the richness of annotation, and the overall amount of curations that can be made \citep[see ][ for an early example]{voormann2008agile}. Hence, curation becomes increasingly complex on the large data sets that have become commonplace.
Rule-based algorithms, regular expressions, and machine learning methods, are powerful tools for improving data quality on large data sets and can correct many thousands of errors simultaneously \citep{khirbat2017ocr,zhang2020spelling, fante2021quantitative, yun2021re} or filter out poor-quality data \citep{dakka2021automated}. However, such approaches do not have accuracy guarantees and usually come with non-negligible error rates, essentially introducing new errors to the data.
Moreover, augmenting or curating a data set with crowdsourcing, non-expert annotators, and external data sources
are also possible approaches to curation. However, these approaches also have a non-negligible amount of errors \citep{barchard2011preventing, pmlr-v48-gaoa16,bruhlmann2020quality}. 
Hence, it is not possible today to curate data at scale efficiently without also risking the introduction of new errors.


\subsection{Our Contributions}

In this work, we address \textit{data accuracy} in large datasets, particularly textual and tabular data, where full manual inspection is infeasible and automated methods must be both scalable and verifiable. 
Our main contributions are:
\begin{enumerate}
    \item \textbf{Iterative curation framework:} A structured process for scalable data accuracy improvements in large textual and tabular datasets. 

    \item \textbf{Theoretical guarantees:} We prove the \emph{asymptotic convergence to an error-free dataset}, the \emph{exponential decay of errors} and a guarantee on the \emph{rate of error reduction} on expectation at each revision of the data.

    \item \textbf{Validation:} We support our results with simulations and a real-world case study on the Swedish Parliamentary Corpus.  
\end{enumerate}


\section{Iterative Curation of Large Data}


The central idea of iterative curation is to incrementally improve data quality over time, as initially described by \citet{voormann2008agile}, extended to the context of large data sets and automated curation. 
We define the dataset as a mutable state that evolves through well-defined revisions and study the theoretical properties of this process.
\begin{figure}[htb]
    \centering
    \includegraphics[width=0.99\linewidth]{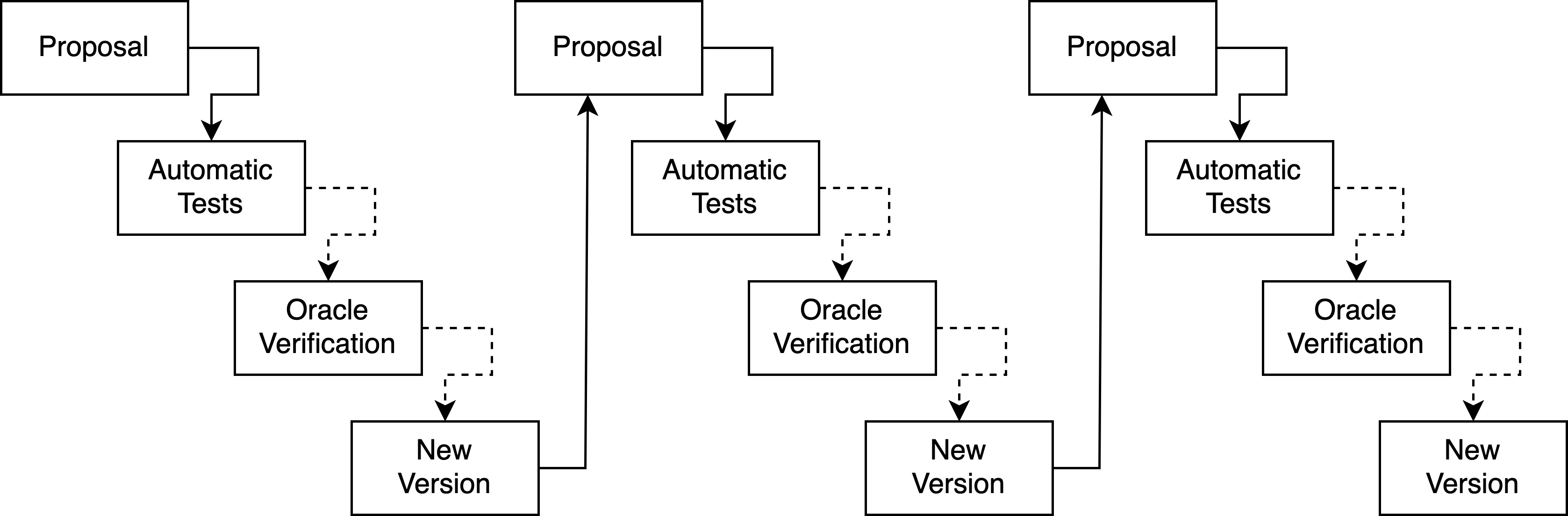}
    \caption{Iterative Data Curation}
    \label{fig:graph-different-proposals}
\end{figure}

The process begins with setting up a prototype data set. After this initial step, the remaining three steps, revision, evaluation, and release, are repeated in short cycles (see Figure \ref{fig:graph-different-proposals}).
To ensure that each iteration leads to measurable improvement, we assume we have access to an oracle to assess whether an individual change, or an \emph{edit}, is correct or true. In practice, the oracle is most likely a human manually assessing whether a change is correct.

This formalism connects straightforwardly to version control, automated testing, and semantic versioning---principles familiar from software engineering---to data. Changes to the dataset, \emph{revisions}, become analogous to git \textit{commits} and each release is documented and versioned. Further, we treat the dataset state as an Application Programming Interface (API). The API defines how users and systems interact with the data: retrieving content, querying metadata, or validating data quality and integrity. Semantic versioning \citep{preston2013semantic} is used to describe the changes to the (data) API by accepted revisions.

\subsection*{Step 1: Set up a Prototype}

The process begins with the creation of a minimal prototype dataset, which serves as the initial state of the data, stored in the data format $\mathcal S$. At this stage, the primary goal is to define a suitable $\mathcal S$---ideally, one that is both human-readable and easily extensible. For example, in the case of a corpus of parliamentary proceedings, a prototype might consist of the full text of a small number of articles encoded in a  ParlaClarin TEI XML schema \citep{erjavec2021parla}. This initial prototype should be small enough to allow thorough manual inspection and validation. Furthermore, the format $\mathcal S$ can consist of multiple formats, such as text documents and metadata files, i.e., $\mathcal S = \mathcal S_1 \times \mathcal S_2$. Once the structure is deemed satisfactory, the initial version of the complete data $S_0 \in \mathcal S$ is created.



\subsection*{Step 2: Create a Revision Proposal} 

Given the initial state $S_0$, subsequent revision proposals, $R_t\in\mathcal{S}$, are introduced as either (1) additions or (2) corrections, each representing a targeted update to the dataset and implying a proposal $R_{t+1}$ for a new data state $S_{t+1}$. Additions mean that the data is expanded in different directions of interest, extending the format $\mathcal S_{new}  := \mathcal S \times \mathcal S^\prime$, e.g., new speeches or metadata on members of parliament are added in a parliamentary proceedings dataset. Conversely, a correction means that the data state is changed, but the format $\mathcal S$ remains unchanged, e.g., when correcting identified errors.
Step 2 results in a new revision proposal $R_{t+1} \in \mathcal S$ from the previous state $S_{t}$.

\subsection*{Step 3: Accept or Reject the Proposal}

Once a proposal $R_{t+1}$ for a new state has been created, the next step is determining whether it should be accepted and applied or rejected. This is done in two steps: (1) automated data tests and (2) (oracle) accuracy control of the proposed revision. The automated tests do not, in general, cost much to run, while the (oracle) accuracy control usually entails costly manual work of assessing the correctness of $R_{t+1}$. Hence, the automated tests are run first, rejecting poor $R_{t+1}$ proposals before the edit sampling and oracle checking.

\emph{Step 3.1: Automatic Data Testing Step}

Automated testing of proposals can take several forms. Some involve format validation, which ensures that the proposal $R_{t+1}$ adheres to the expected structural specifications; in other words, that $R_{t+1} \in \mathcal S$. This may include validating XML, checking that tables contain all required columns, and confirming that no files have been unintentionally removed or added. 
A second example is sanity checks of content, which evaluate whether specific properties of the data are preserved based on prior knowledge. These may include verifying known summary statistics---such as the total number of members of parliament,
or checking that previously fixed errors or corrections have not reappeared. 





\emph{Step 3.2 Edit Sampling Step}

Suppose the new revision passes the automatic data tests. In that case, we control the quality of a new revision by taking a random sample of the individual \emph{edits}, $\Delta_t = \Delta(R_{t+1}, S_t) = u_1, \dots, u_{\lvert \Delta_t \rvert}$. As the sample is drawn from the set of edits, there needs to be a specific definition of edits for each data format $\mathcal S$. For tabular data, we propose uniform sampling from the set of changed entries (i.e. rows); for sequential data, we propose sampling unit additions and deletions from the addition-deletion diff, obtained via the Myers algorithm \citep{myers1986nd}, e.g. using git \texttt{diff} \citep{git} (see Appendix \ref{appendix:diffsampling} for details).

The sample of edits $u_i$ of size $n$ is checked through the oracle to assess if a proposed edit $u_i$ is correct or not. 
A revision is accepted, i.e. $S_{t+1} := R_{t+1}$, if the proportion of edits deemed correct by the oracle is sufficiently high, i.e. over a threshold $m \in \{1, \dots, n\}$; otherwise the proposed revision is rejected and $S_{t+1} := S_t$.


\subsubsection*{Step 4: Release a new version}

Finally, if a revision $R_{t+1}$ has been accepted, we release a new state, $S_{t+1}$, of the data with the amended revision.
Each release of data states, $S_{t+1}$, is associated with a version number which reflects the type of change compared to the previous version, $S_{t}$. We adapt the semantic versioning by \citet{preston2013semantic} for data as follows:
\begin{itemize}
    \item MAJOR version is increased when changes are made that are not backwards compatible for the API of the data ($\mathcal S$ has changed),
    \item MINOR version when adding functionality, e.g. as new metadata fields or text documents in a backwards-compatible manner ($\mathcal S$ has been extended $\mathcal S^\prime := \mathcal S \times \mathcal S_{new}$), and
    \item PATCH version when you perform fixes or changes that do not introduce new features in the data, e.g. fixing errors in data or adding new data tests ($\mathcal S$ remains unchanged).
\end{itemize}
Similar to standard semantic versioning, only major and minor versioning is followed before the 1.0.0 release and the data is deemed stable upon the 1.0.0 release.

\section{Theoretical Guarantees} \label{section:theoretical_properties}


To study the theoretical properties of the data accuracy under the proposed iterative curation process, we must make assumptions about the data format, the true data and the proposed changes. The analysis relies on the following assumptions
\begin{enumerate}[label=(\alph*)]
    \item There is a canonical form $\mathcal S$ of the data, i.e. there is
    a one-to-one mapping from what the data represents to its representation $S \in \mathcal S$. Further, there is a distance $d: \mathcal S \times \mathcal S \to \mathbb N$ in the space $\mathcal S$, and access to the set of edits $\Delta(S, S^\prime)$ (where $\lvert \Delta(S, S^\prime) \rvert = d(S, S^\prime)$) between any two versions or revisions $S, S^\prime$. The distance to the true data $S^\star$ defines the number of errors $E_t := d(S^\star, S_t)$.\label{assumption:canonical}
    \item The set of edits $\Delta_t = \Delta(R_{t+1}, S_t)$ between the proposal $R_{t+1}$ and the current state consists of unit modifications $u_1, \dots, u_{\lvert \Delta_t \rvert}$, which each modify the distances by $\{-1, 0, 1 \}$, and are additive.
    in terms of distance: $d(S_t + u + u^\prime, S^\star) - d(S_t, S^\star) = d(S_t + u, S^\star) - d(S_t,S^\star) + d(S_t + u^\prime, S^\star) - d(S_t, S^\star)$ for all possible edits $u \neq u^\prime$. \label{assumption:additivity}
    \item There exists an oracle that can calculate the changes in the distances $d(S_t + u_i, S^\star) - d(S_t,S^\star)$ caused by any single edit $u_i, \forall i \in \{1, \dots, \lvert \Delta_t \rvert \}$. \label{assumption:sample}
    \item The revision size $\lvert \Delta_t \rvert$ is binomially distributed with $E_t$ trials and the success probability $\lambda_t$, where $\lambda_t$ is i.i.d over all timesteps $t$. \label{assumption:revsize}
    \item The number of correct edits $\sum_{u \in \Delta_t} \mathbb I (d(R_{t+1}, S^\star) - E_t = -1)$ for proposals $R_{t+1}$ is binomially distributed with the accuracy $1-r_t$. The proposal error rate $r_t$ is distributed i.i.d. over all timesteps $t$, independently of $\lambda_t$, and has a nonzero density $p(r_t)$ whenever $r_t < 0.5$, bounded from below by $a>0$.\label{assumption:revacc}
\end{enumerate}



\begin{theorem} 
\label{theorem:zeroerrors}
Given Assumptions \ref{assumption:canonical}-\ref{assumption:revacc}, when only proposals with at least $m$ correct modifications out of a sample of $n$ are accepted, provided that the following holds for $(n,m)$
\begin{equation}
    a > \frac{C(n,m)}{2 + C(n,m)}
\end{equation}
where $C(n,m)$ is defined as
\begin{equation} \label{eq:c_function}
\begin{aligned}
    C(n,m) &= \frac{2 (n+1)(n+2)}{2m - n} {n \choose m} \\
    &\cdot (r_{max} - 1/2)(1-r_{max})^mr_{max}^{n-m} \\
    r_{max} & = \frac{3n - 2m  +2 + \sqrt{(2m-n)^2 + 4(n + 1)}}{4(n + 1)} 
\end{aligned}
\end{equation}
the number of errors $\lim_{t \to \infty} P(E_t = 0) \to 1$.
\end{theorem}

\begin{proof}
Given Assumptions \ref{assumption:canonical}--\ref{assumption:additivity}, all changes observed in a proposal $R_{t+1}$ can be enumerated, and applying them in any order will yield a difference in distance to the true dataset corresponding to the sum of the distances of them being separately applied to $S_t$. Thus, we can estimate the difference made in the distance $d(R_{t}, S^\star) - d(S_{t}, S^\star)$ by taking a simple random sample (SRS) $u_1, \dots u_n$ of size $n$ from $\Delta(R_{t+1}, S_t)$, and obtaining the distances $d(S_t + u_i, S^\star) - d(S_t, S^\star) \quad \forall i \in \{1, \dots , n\}$ by Assumption \ref{assumption:sample}. The number of correct edits in the sample $M$ is
\begin{align}
M &= \sum_{i=1}^n \mathbb{I}(d(S_t + u_i, S^\star) - d(S_t, S^\star) = -1)
\end{align}
which is binomially distributed $M \sim Bin(r_t, n)$. The posterior distribution of $r_t$ is thus proportional to $p(M = m \mid r_t)p(r_t) \propto (1-r_t)^m r_t^{n-m}p(r_t)$, where $p(r_t)$ is the prior PDF of $r_t$.

By assumption \ref{assumption:revsize}, $\lvert \Delta_t \rvert$ is independent of the revision error rate $r_t$. 
Moreover, as it is binomially distributed $\lvert \Delta_t \rvert \sim B(\lambda_t, E_t)$ conditionally on $\lambda_t$, it can be decomposed to a sum of $E_t$ independent Bernoulli random variables.
By Assumption \ref{assumption:revacc}, the errors given an $\lvert \Delta_t \rvert$ are binomially distributed. $E_{t}$ can thus be represented as a branching process with variable offspring probabilities $p_{0,t}, p_{1,t}, p_{2,t}$ that are only a function of $\lambda_t$ and $r_t$. We obtain a branching process in random environments (BPRE) \citep{smith1969branching}, where $p_{0,t} = (1-r_t) \lambda_t$, $p_{1,t} = 1 - \lambda_t$ and $p_{2,t} = r \lambda_t$. 

A BPRE $Z_t \in \mathbb{N}$ where $Z_0 > 0$ is defined as a Markov process where $Z_{t+1}$ is a sum of $Z_t$ i.i.d. random variables $\zeta \in \mathbb N$ with probabilities $p_{0, t}, p_{1, t} \dots$ that are parametrized by $\theta_t$ which are i.i.d for each $t$. Moreover, $\lim_{t \to \infty} P(Z_t=0) = 1$ for all BPREs that have $\mathbb{E}_\theta [\log \mathbb{E}[\zeta]] < 0$ \citep{smith1969branching}. In our case, $\zeta$ is parameterized by $r_t$ and $\lambda_t$. $M$ follows a binomial distribution, thus $m < m^\prime \implies P(r_t \leq r \mid M = m) < P(r_t \leq r \mid M = m^\prime) \quad \forall r \in (0,1)$, and $\mathbb{E}_{r_t}(\mathbb{E}[\zeta] \mid M = m])$ is thus a decreasing function of $m$.
Therefore,
\begin{align}
        \mathbb{E}_{r_t, \lambda_t}[\mathbb{E}[\zeta] \mid M = m] &< 1  \\
        \implies \quad \mathbb{E}_{r_t, \lambda_t}[ \log \mathbb{E}[\zeta] \mid M = m] &< 0 \label{eq:jensen} \\
        \implies \quad \mathbb{E}_{r_t, \lambda_t}[ \log \mathbb{E}[\zeta] \mid M \geq m] &< 0
\end{align}
where step (\ref{eq:jensen}) uses Jensen's inequality.
Consider the expectation $\mathbb{E}_{r_t, \lambda_t}[\mathbb{E}[\zeta] \mid M = m]$
\begin{equation} \label{eq:zeta_expectation}
\begin{aligned}
    &\mathbb{E}_{r_t, \lambda_t}[\mathbb{E}[\zeta] \mid M = m] \\
    &= \mathbb{E}_{r_t, \lambda_t}\Big[\sum_{k=0}^2 k \cdot p_{k,t}  \mid M = m\Big]\\
    &=1 - 2 \mathbb{E}[\lambda_t] \left( 1/2  - \mathbb{E}[r_t \mid M = m] \right)
\end{aligned}
\end{equation}
which is less than 1 whenever the $\mathbb{E}[r_t \mid M= m] < 1/2 \iff \mathbb{E}[1/2-r_t \mid M= m]>0$, 
%
i.e.
\begin{align}
    & \int_0^1 (1/2 - r) p(r \mid m) d r \\
    & \propto \int_0^{0.5} (1/2 - r) p(m \mid r) p(r) d r \\
    &+ \int_{0.5}^1 (1/2 - r)p(m \mid r)  p(r) d r \nonumber \\
    &\geq a \Big(\frac{2m - n}{2(n+1)(n+2)}\Big) \\
    &- \max_r {n \choose m}\left[(r-1/2)(1-r)^mr^{n-m}\right] (1-  a/2) \nonumber \\
    &\propto a - \Bigg(\frac{2 (n+1)(n+2)}{2m - n} {n \choose m} \cdot \\
    &\quad \quad \quad \quad \max_r \left[(r-1/2)(1-r)^mr^{n-m}\right] (1-  a/2)\Bigg) \nonumber\\ 
    &\propto a - C(n, m)/(2 + C(n, m)) > 0
\end{align}
See Appendix \ref{appendix:detailed_proof} for a detailed derivation. The maximum of $(r-1/2)(1-r)^mr^{n-m}$ is found at 
\begin{equation} \label{eq:r_max}
\begin{aligned}
    r_{max} &= {\arg \max}_r (r-1/2)(1-r)^mr^{n-m} \\
    &= \frac{3n - 2m  +2 + \sqrt{(n-2m)^2 + 4(n + 1)}}{4(n + 1)} 
\end{aligned}
\end{equation}
and thus $C(n,m)$ becomes
\begin{equation} \label{eq:c_function}
\begin{aligned}
    C(n,m) = \Bigg( &\frac{2 (n+1)(n+2)}{2m - n} {n \choose m} \\ &\cdot(r_{max} - 1/2)(1-r_{max})^mr_{max}^{n-m} \Bigg)
\end{aligned}
\end{equation}
Details for the derivation of $r_{max}$, as well as bounds for the value of $C(n,m)$, are provided in Appendix \ref{appendix:derivations}. Finally, the condition for the extinction of $E_t$ becomes
\begin{equation}
    \begin{aligned}
        a > \frac{C(n, m)}{2 + C(n, m)}
    \end{aligned}
\end{equation}
which can be attained for any $a$ as $C(n,m)$ can be made arbitrarily small by increasing $m$ and $n$ (details in Appendix \ref{appendix:derivations}).\end{proof}

\begin{remark}
A subcritical BPRE decays exponentially, both in terms of the mean $\mathbb{E}[Z_t]$ of the process by timestep $t$ \citep{smith1969branching} and an upper bound for the survival probability $P(Z_t >0)$ \citep{guivarc2001proprietes}. The decay happens wrt. the mean of the offspring probabilities $\mu = \mathbb{E}[\zeta]$, i.e. $E_t \sim \mu^t$, while $\mu$ is
determined by the posterior distribution of $r_t$ and $\lambda_t$.
%
Specifically, we have $\mu = 1 - 2 \mathbb{E}[\lambda_t](1/2 - \mathbb E[r_t \mid M \geq m]) $ given acceptance. In other words, large and accurate revisions speed up the decay of errors. 
\end{remark}

\begin{remark}
The optimal value of $m$ depends on $n$ and the distribution of $r_t$, and tends to $n/2$ as $n \to \infty$. For small $n$, the optimal $m \in \{n/2 +1, \dots, n\}$ can  be different, as seen in Figure \ref{fig:optimal_m}.
\end{remark}

\begin{figure}[h!]
    \centering
    \begin{subfigure}[b]{0.45\textwidth}
    \includegraphics[width=\linewidth]{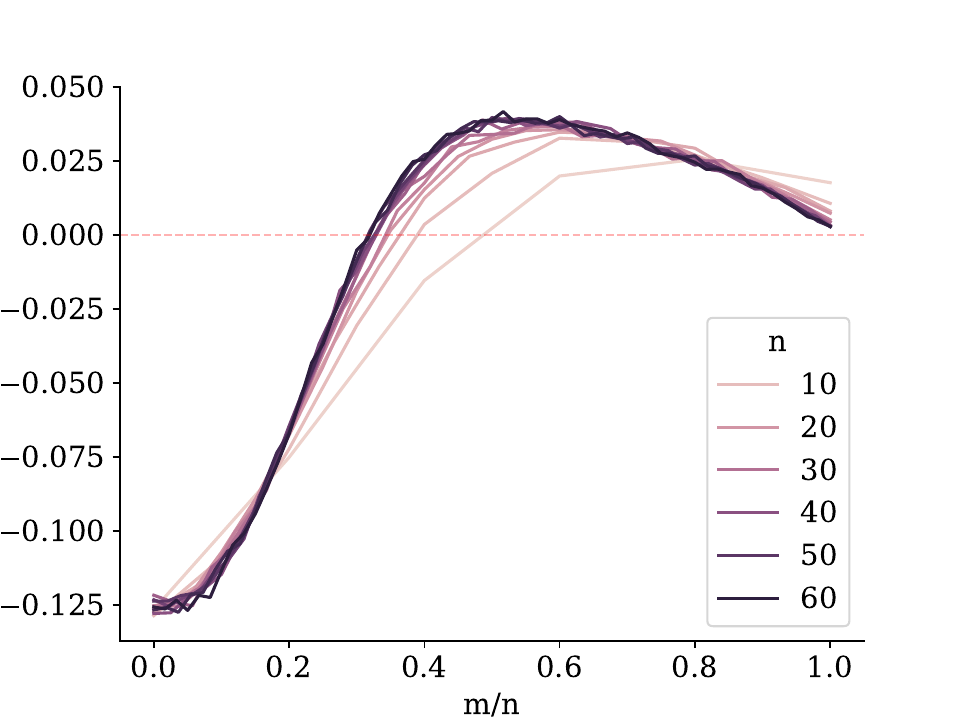}
    \end{subfigure}
    \hfill
    \begin{subfigure}[b]{0.45\textwidth}
    \includegraphics[width=\linewidth]{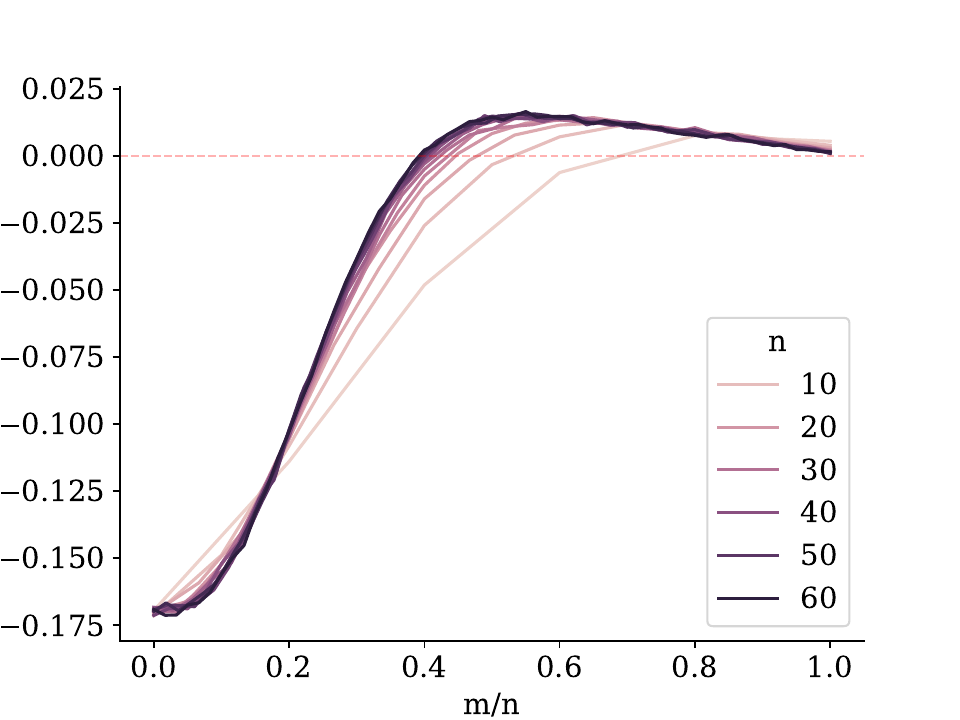}
    \end{subfigure}
    \caption{Improvement rate $-\log \mathbb{E}[E_{t+1}/E_t]$, by the acceptance threshold $m/n$. $P(r_t \leq 0.5) = 0.15$ on the left, $P(r_t \leq 0.5) = 0.05$ on the right. Values above the dashed line correspond to convergence, i.e. $E_t \overset{p}{\to} 0$.}
    \label{fig:optimal_m}
\end{figure}

\begin{theorem} \label{theorem:noisy_oracle}
    If we have access to a noisy oracle that with probability $1-\varepsilon$ observes the oracle, and with probability $\varepsilon$ labels any edit as correct, \ref{assumption:canonical}, \ref{assumption:additivity}, \ref{assumption:revsize} and \ref{assumption:revacc} of Theorem \ref{theorem:zeroerrors} hold, and we set the acceptance threshold to $m_{noisy} := m/(1-\varepsilon)$, the number of errors goes to zero, i.e. $\lim_{t \to \infty}P(E_t=0)=1$.
\end{theorem}
\begin{proof}
A noisy oracle yields the following posterior
\begin{equation} \label{eq:noisy_posterior}
\begin{aligned}
p(r \mid M=m) \propto ((1-\varepsilon)(1-r) + \varepsilon)^m ((1-\varepsilon)r)^{n-m} \\
 \propto ((1-\varepsilon)(1-r) + \varepsilon)^m r^{n-m}
\end{aligned}
\end{equation}

Due to Lemma \ref{lemma:alt_posterior} presented in the Appendix, this posterior stochastically dominates the following posterior
\begin{equation} \label{eq:noise_free_posterior}
\begin{aligned}
p(r \mid M=m) &\propto  (1-r)^{(1 - \varepsilon)m} r^{n-m} \\
&=  (1-r)^{(1 - \varepsilon)m} r^{(n + \varepsilon m) - (1 - \varepsilon)m} \\
&=  (1-r)^{m_{noisy}} r^{n_{noisy} - m_{noisy}} \\
\end{aligned}
\end{equation}
which in turn corresponds to the posterior of Theorem \ref{theorem:zeroerrors} obtained with $m_{noisy} := m (1-\varepsilon), n_{noisy} := n + \varepsilon m$. Inverting the mapping, noisy observations correspond to $m = m_{noisy}/(1-\varepsilon)$ and $n = n_{noisy} - \frac{\varepsilon}{1-\varepsilon} m_{noisy}$. The stochastical dominance means that the expectation of $r$ with the posterior of Eq. (\ref{eq:noisy_posterior}) is smaller than with the posterior of Eq. (\ref{eq:noise_free_posterior}), which means the expectation of Eq. (\ref{eq:zeta_expectation}) is going to be below $1$ given the assumptions \ref{assumption:canonical}, \ref{assumption:additivity}, \ref{assumption:revsize} and \ref{assumption:revacc} of Theorem \ref{theorem:zeroerrors}.
\end{proof}

From Theorem \ref{theorem:noisy_oracle} we can see that we can easily account for a noisy oracle. This is , as expert knowledge perfect.

In Theorem \ref{theorem:zeroerrors}, we do not specify the data format $\mathcal S$. On the other hand, assumptions \ref{assumption:canonical} and \ref{assumption:additivity} clearly depend on the data format. For that reason, we want to verify them for two common data formats: fixed-sized tables of arbitrary values, and sequences of text. For the former, the Hamming distance is a natural metric as corresponds to the number of incorrect values; for the latter, the addition-deletion edit distance (\textit{Levenshtein distance} without substitutions) is a natural choice as it tells the number of edits needed to get to the correct data $S^\star$.
\begin{theorem} 
\label{theorem:tabulardata}
Tabular data, using the Hamming distance $d(S, S^\prime) = \sum_{i=1}^N \mathbb{I} (s_i \neq s_i^\prime)$, fulfills the assumptions \ref{assumption:canonical}-\ref{assumption:revacc} of Theorem \ref{theorem:zeroerrors}.
\end{theorem}
\begin{proof}
Tabular data consists of a matrix $A \in \mathbb{R}^{N \times L}$. This is the data format of Assumption \ref{assumption:canonical}, for which the suitable distance is the number of non-equal entries, ranging from $0$ when the tables are equal to $L \cdot N = K$ when each entry is different. 

Without loss of generality, we can consider vectors $v \in \mathbb{Z}^K$ equipped with the metric $d(v, v^\prime) = \sum_{j=1}^K \mathbb{I}(v_j \neq v_j^\prime)$.
All possible unit modifications $u$ are integer-scaled one-hot vectors, so $u = z \cdot [0, \dots, 0, 1, 0, \dots, 0]^T$. These changes are orthogonal; two non-orthogonal changes constitute one change. We thus can see that they are also additive in terms of the distances for Assumption \ref{assumption:additivity} (see Appendix \ref{appendix:detailed_proof} for details).

The remaining assumptions do not concern the data structure and are fulfilled for tabular data.
\end{proof}

For Theorem \ref{theorem:zeroerrors} to hold for sequential data, such as textual data, we need further assumptions
\begin{enumerate}[label=(\alph*)]
    \setcounter{enumi}{5}
    \item The true sequence $S^\star$ only contains distinct elements, i.e. $S^\star = [s^\star_1, \dots , s^\star_N]$, and $i\neq j \implies s^\star_i \neq s^\star_j \quad \forall (i,j) \in \{1, \dots, N\}$ \label{assumption:distinct}
    \item Each version $S_t$ only contains distinct elements \label{assumption:duplicated_addition} 
    \item No version $S_t$ contains any two elements $(s_i^\star, s_j^\star)$ of $S^\star$ in incorrect order \label{assumption:permutation} 
\end{enumerate}

\begin{theorem}
\label{theorem:sequencedata}
Sequences of text, using the addition-deletion edit distance, fulfil the assumptions \ref{assumption:canonical}-\ref{assumption:revacc} of Theorem \ref{theorem:zeroerrors} given further assumptions \ref{assumption:distinct}-\ref{assumption:permutation}. 
\end{theorem}

\begin{proof}
The addition-deletion edit distance is defined for any two finite sequences of countable elements. It belongs to $\mathbb N$, i.e. $\mathcal S$ is equipped with a suitable metric $d$ for Assumption \ref{assumption:canonical}.
Furthermore, the edit distance can be obtained via the longest common subsequence (LCS)
\begin{equation}
    d(S, S^\prime) = \lvert S \rvert + \lvert S^\prime \rvert - 2 \lvert LCS(S, S^\prime) \rvert 
\end{equation}

To verify Assumption \ref{assumption:additivity} for any two edits $(u_1,u_2)$, there are three possible cases: two additions, two deletions and an addition and a deletion. If the LCS is unique, we can inspect these cases via the effects of all potential pairs $(u_1, u_2)$ on the LCS. This holds for sequences consisting of distinct elements that do not contain reordered elements of each other (proof in Appendix \ref{appendix:detailed_proof}), as guaranteed by Assumptions \ref{assumption:distinct} and \ref{assumption:permutation}.

In the case of two erroneous additions, regardless of the order or position of these additions, $E_t$ always increases by 2, given that the true data $S^\star$ and $S_t$ only have distinct elements, guaranteed by Assumptions \ref{assumption:distinct}--\ref{assumption:duplicated_addition}. This is because additions do not change the longest common substring if they are distinct from all elements of $S^\star$.
Similarly, two correct additions are going to yield the same decrease of $-2$ to the distance, since they both increase the length of the LCS by 1, given that their order is not inverted. This is guaranteed by Assumption \ref{assumption:permutation}. Thus in the case any two additions $u_1, u_2$, $d(S_t + u_1 + u_2, S^\star) - d(S_t, S^\star) = d(S_t + u_1, S^\star) - d(S_t, S^\star) + d(S_t + u_2, S^\star) - d(S_t, S^\star)$.

Two deletions act independently, as it does not matter which one is done first; both either increase the edit distance by one if they are incorrect, or decrease it by one if they are correct, and the sum of their effects is the effect of their sums. Thus, $d(S_t + u_1 + u_2, S^\star) - d(S_t, S^\star) = d(S_t + u_1, S^\star) - d(S_t, S^\star) + d(S_t + u_2, S^\star) - d(S_t, S^\star)$ for any two deletions $u_1$ and $u_2$.

The third and final case involves the interactions of additions and deletions. The deletion and addition of two different elements do not interact, and thus, they are additive in terms of distance.
The addition and subsequent deletion of the same element cancel each other and are not present in any deltas. Thus, additivity holds for all possible edits.

The remaining assumptions in Theorem \ref{theorem:zeroerrors} do not concern the data structure, and can also be fulfilled for text sequences.
Thus, text sequences fulfil the assumptions of Theorem \ref{theorem:zeroerrors}.
\end{proof}

\begin{remark}
$\mathcal S$ must be defined so that Assumptions \ref{assumption:distinct}--\ref{assumption:duplicated_addition} in Theorem \ref{theorem:sequencedata} are reasonable. One way to do this is to organise a sequence into larger chunks. If a random sequence is split into chunks of $K$ characters (or words), the probabilities of its composite elements can be made arbitrarily small. For example, chunks of $K$ binary elements have $2^K$ possible categories, and their probabilities decrease exponentially as $K$ is increased, guaranteeing that the elements in the sequence are distinct.
\end{remark}




\begin{theorem} 
For tabular data, the number of errors $E_t^\prime$ with the tests converges faster than the number of errors without the tests $E_t$, i.e. $P(E_t=0 \mid E_0 = E) < P(E_t^\prime=0 \mid E_0^\prime = E) \quad \forall t, E \in \mathbb N$.
\label{theorem:tests}
\end{theorem}
\begin{proof}
Without loss of generality, we consider $K$-dimensional vectors $v \in \mathbb Z^{K}$ of integers as $\mathcal S$. Let the data tests be a set of tuples $W = \{w_1, \dots w_J\}, 0 \leq J \leq K, \quad w_j \in \{1, \dots, K\} \times \mathbb Z$. The tests then \textit{pass} when $\sum_{j=1}^J \mathbb I(v_{w_{j1}} \neq w_{j2}) = 0$.

For each sample, $J^\prime, 0 \leq J^\prime \leq J$, the changes happen in the dimensions where there is a data test. Furthermore, $J^*, 0 \leq J^* \leq J^\prime$ of these tests are both in dimensions where changes happen and not included in the sample. Moreover, $P(J^\prime > 0) > 0$ for all timesteps $t$ unless $E_t = 0$, and subsequently, since sampling is independent of the test $P(J^* > 0) > 0$ if $E_t \neq 0$. If only revisions where the tests pass are accepted, the likelihood is proportional to $B(J^* + m, n-m)$ instead of $B(m, n-m)$.

Let $p^\prime(r_t \mid m)$ and $p(r_t \mid m)$ be the posteriors of the error rates corresponding to $E_t^\prime$ and $E_t$, respectively. Since the posterior $p^\prime(r \mid m) \propto (1-r)^{J^*} p(r \mid m)$, the posterior expectation $\mathbb E[p^\prime(r_t)]$ is smaller than the posterior expectation $\mathbb E[p(r_t)]$.
Moreover, $E_t-E_{t+1}$ and $E_t^\prime-E_{t+1}^\prime$ are binomially distributed with the rate parameter $r_t \lambda_T$, satisfying the monotone likelihood ratio property. This means the $E_{t+1}^\prime$ is stochastically smaller than $E_{t+1}$ for all timesteps given any $E_t = E_{t}^\prime > 0$, i.e. $P(E_{t+1}^\prime \leq x) > P(E_{t+1} \leq x) \quad \forall x \in \mathbb N$, including $x = 0$ which denotes convergence.
\end{proof}

Theorem \ref{theorem:tests} illustrates the importance of adding tests as step 3.1 in the Proposal accept/reject step. By adding automated tests we will automatically reject poor proposals and speed up the convergence.


\section{Experiments}

To assess the proposed iterative curation framework, we perform two simulation studies, which allow us to study the convergence behaviour of the method under varying assumptions. Moreover, we demonstrate the practical applicability of our approach via a use case involving the Swedish Parliamentary Corpus, a large dataset with both textual and tabular data \citep{yrjanainen2024swedish}. The code for the experiments is available in the Supplementary material. 

\subsection{Simulation Study}

First, we simulate the number of errors $E_t$ directly according to Theorem \ref{theorem:zeroerrors}, and set the  error rate $r_t \overset{i.i.d.}{\sim} Beta(\alpha, \beta)$, where $\mathbb{E}(r_t) = \alpha / (\alpha + \beta)$.
The results are presented in Figure \ref{fig:theoretical_simulation}. As expected from a BPRE process, the decay of the errors is exponential. Moreover,  the threshold $m/n = 0.6$ outperforms $m/n = 0.5$ when $n=10$, while $m/n = 0.5$ outperforms $m/n = 0.6$ when $n=50$ (see Appendix \ref{appendix:figures}).

\begin{figure}[h!]
    \centering
    \begin{subfigure}[b]{0.45\textwidth}
    \includegraphics[width=\linewidth]{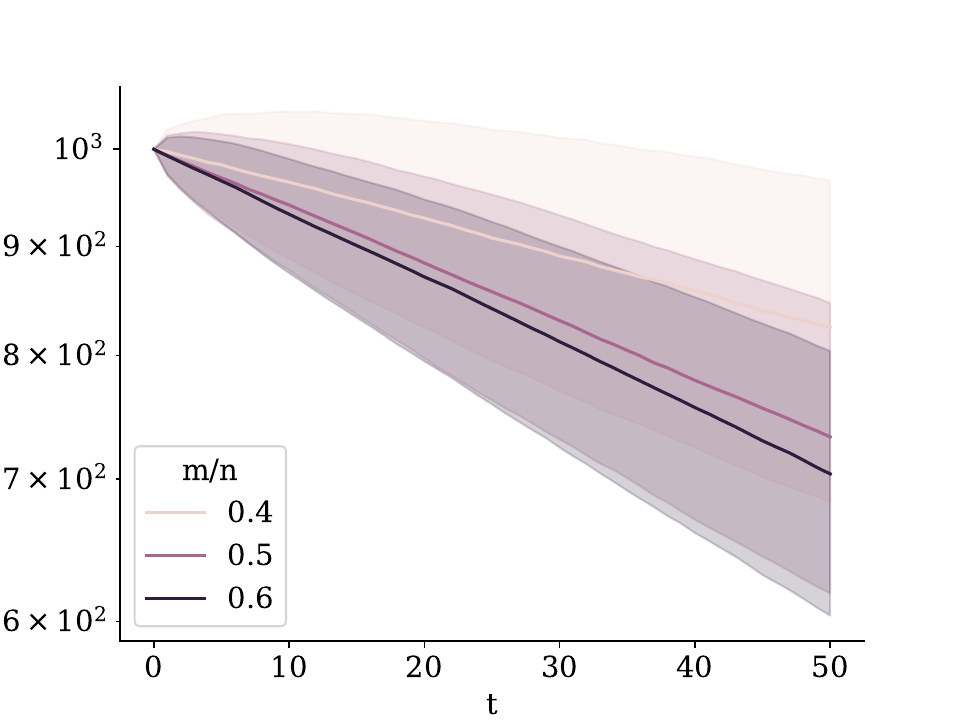}
    \caption{$r_t \sim Beta(2,1)\,,\mathbb{E}(r_t)=2/3$}
    \label{fig:theoretical_simulation}
    \end{subfigure}
    \hfill
    \begin{subfigure}[b]{0.45\textwidth}
    \includegraphics[width=\linewidth]{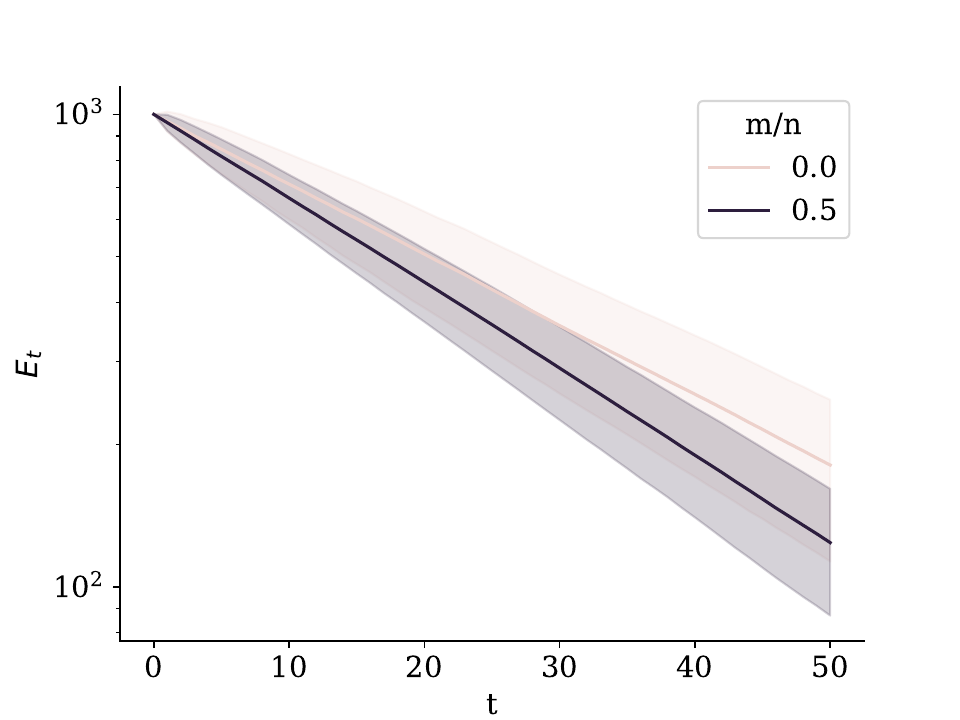}
    \caption{$r_t \sim Beta(1,2)\,,\mathbb{E}(r_t)=1/3$.}
    \label{fig:theoretical_simulation_goodproposals}
    \end{subfigure}
    \caption{Number of errors $E_t$ in the data as a function of $t$, $n=10$. The exponential decay of the errors is seen on the logarithmic $y$ axis.}
\end{figure}
Furthermore, Figure \ref{fig:theoretical_simulation_goodproposals} shows that even if the proposals are good enough to guarantee convergence, using the decision rule speeds it up. This is the case both for $n=10$ and $n=50$, where $m/n = 0.5$. We found this to be the case even with highly optimistic  $r_t$, such as $Beta(1,4)$, where 94\% of the proposals would improve the data (see Appendix \ref{appendix:figures}). In addition, a simulation with a noisy oracle was done, where the errors converged similarly to a matching process with a true oracle, in line with Theorem \ref{theorem:noisy_oracle} (see Appendix \ref{appendix:figures}).

As a second simulation, we generate a sequence $S^\star$ of 1,000,000 words sampled from the empirical distribution of the English Wikipedia.
It includes lowercase words, spaces, and line breaks.
\begin{figure}[tbh!]
    \centering
    \begin{subfigure}[b]{0.450\textwidth}
    \hspace*{-0.5cm}
    \includegraphics[width=\linewidth]{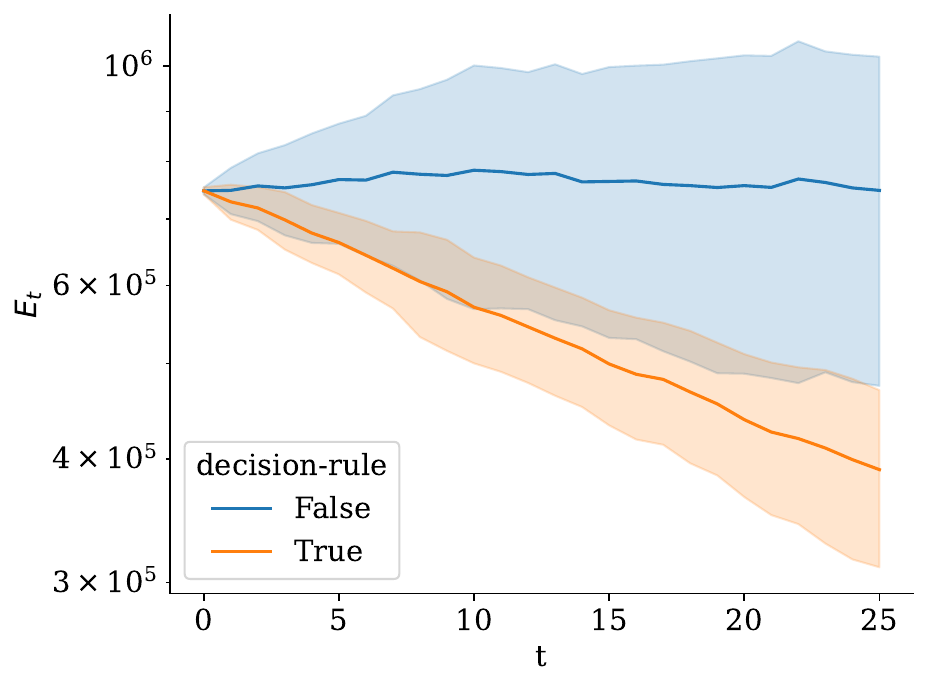}
    \caption{$r_t \sim Beta(1,1)\,,\mathbb{E}(r_t)=1/2$}
    \end{subfigure}
    \hfill
    \begin{subfigure}[b]{0.450\textwidth}
    \hspace*{-0.5cm}
    \includegraphics[width=\linewidth]{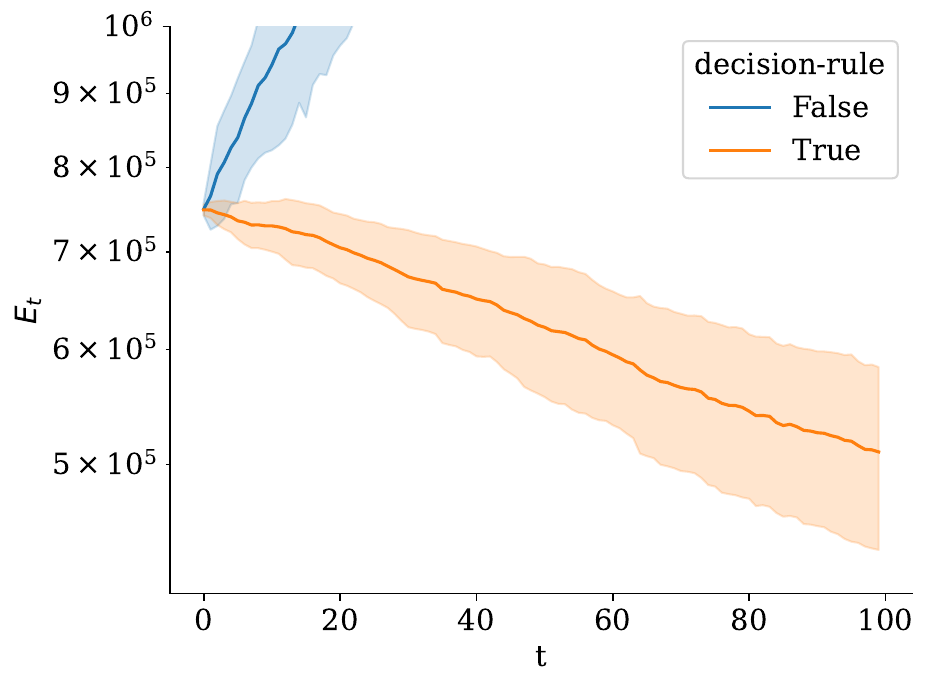}
    \caption{$r_t \sim Beta(5,3)\,,\mathbb{E}(r_t)=5/8$.}
    \end{subfigure}
    \caption{Errors in the simulated text dataset. Decision rule in orange (using $n=50, m=25$), no decision rule in green. 50 simulation runs, or until $E_t=10^6$. }
    \label{fig:text_simulation}
\end{figure}
To create an imperfect copy of $S^\star$ to be curated, we remove, add and change words with probability $1/300$ each.
This results in an initial data set with roughly 10,000 errors. For the deltas, we use the Myers difference algorithm on the word level.
Figure \ref{fig:text_simulation} shows the evolution of $E_t$ in the simulated text data over time. With the decision rule $m/n=0.5$, $E_t$ decreases exponentially, both in terms of the mean and the quantiles. This holds for different distributions of $r_t$ (see Appendix \ref{appendix:figures} for details). Without the decision rule, when $r_t \sim Beta(1,1)\,,\mathbb{E}(r_t)=1/2$, the mean stays static and variance increases with $t$. When $r_t \sim Beta(5,3)\,,\mathbb{E}(r_t)=5/8$, $E_t$ increases exponentially.

\subsection{Case Study: The Swedish Parliamentary Records}
\label{section:case_study}


As a practical use case, we present the curation of the Swedish parliamentary proceedings. In the corpus
, the iterative approach has been used, and proposals are evaluated against automated testing and a random sample of 50 edits \citep[see][for details]{yrjanainen2024swedish}. The  material is crucial for research in numerous scientific fields and Swedish parliamentary proceedings is utilised in training Swedish encoder models \citep{malmsten2020playing}, audio-to-text models, and Nordic LLMs \citep{ohman2023nordic}.
The data comprises 17,938 parliamentary records between 1867 and 2024 with over 500 million words, in ParlaClarin XML format \citep{erjavec2021parla}, and metadata in a CSV format.

Mapping each speech in the text to the speaker in the metadata database ("speaker mapping") is a quality metric of particular importance to applied researchers in the social sciences.
An introduction is marked as 'unknown' if we can identify a speech but not map it to an identified MP.
Figure \ref{fig:swerik_mapping} shows the number of unknown speakers over time, i.e. the reduction of missing annotations in the data. This accuracy has consistently improved over time (timestep $2$ shows an increase due to the removal of false positives).

\begin{figure}[h!]
    \centering
    \begin{subfigure}[b]{0.450\textwidth}
    \hspace*{-0.6cm}
    \includegraphics[width=\linewidth]{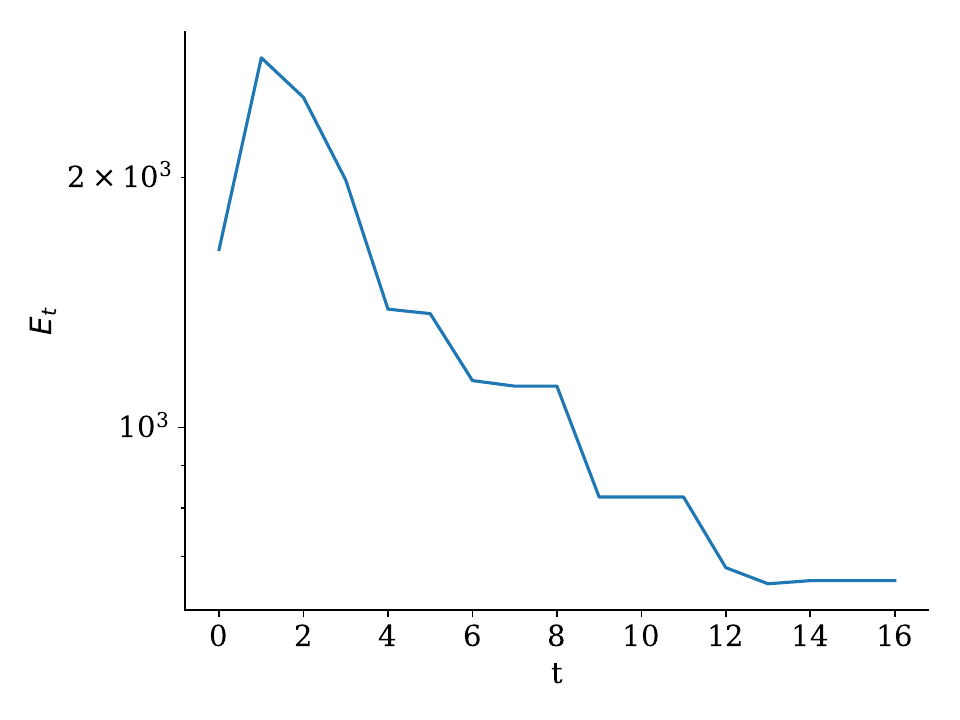} 
    \caption{Speaker mapping errors}
    \label{fig:swerik_mapping}
    \end{subfigure}
    \hfill
    \begin{subfigure}[b]{0.450\textwidth}
    \hspace*{-0.6cm}
    \includegraphics[width=\linewidth]{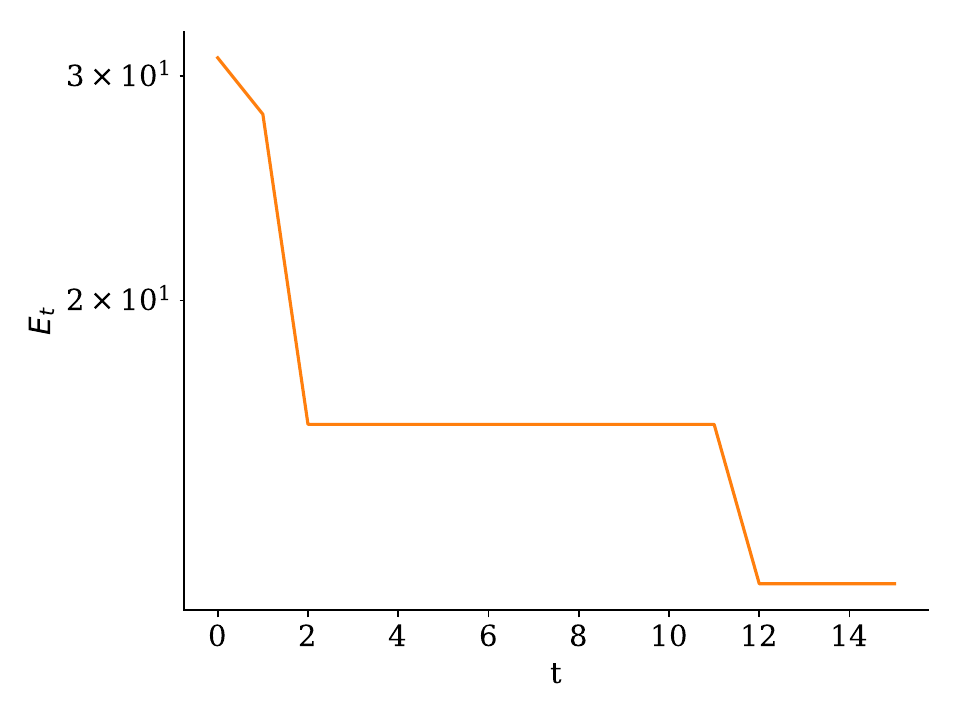}
    \caption{Paragraph classification errors}
    \label{fig:swerik_noteseg}
    \end{subfigure}
    \caption{Errors over time in the Swedish Parliament data. The $y$ axis is logarithmic to show the exponential decay of errors.}
    \label{fig:swerik}
\end{figure}

In addition to speaker introductions, the records contain transcriptions of speech, as well as non-speech elements such as transcriber notes, titles and margin notes. All paragraphs are classified into these classes ("paragraph classification"). 
It is not trivial to know which paragraphs are transcriptions of speech---the part of data that is of interest to most users---and thus they need to be programmatically detected due to the size of the corpus. The accuracy of this classification process is another data accuracy dimension of interest. Figure~\ref{fig:swerik_noteseg} shows the progression of the errors in this classification over the course of the corpus curation. The errors have either been reduced or remained the same for each new version.

\section{Conclusion}

We propose a scalable framework for iterative dataset data accuracy improvements. It supports iterative, high-quality updates with minimal manual overhead by
evaluating revisions via automated tests and targeted sampling. 
Furthermore, we provide theoretical guarantees for the exponential decay of errors and  convergence to an error-free dataset, the first such results to the best of our knowledge. Finally, we demonstrate the method’s effectiveness via simulations and a real-world case study. Our approach is especially suited for large, evolving datasets where data accuracy is of direct importance, with immediate impact on robust social science research, machine learning benchmarks and high-quality LLM training corpora.

\vspace{0.4cm}

\newpage




\bibliographystyle{unsrtnat}
\bibliography{process_references}

\section*{Checklist}

\begin{enumerate}

  \item For all models and algorithms presented, check if you include:
  \begin{enumerate}
    \item A clear description of the mathematical setting, assumptions, algorithm, and/or model. [Not Applicable]
    \item An analysis of the properties and complexity (time, space, sample size) of any algorithm. [Not Applicable]
    \item (Optional) Anonymized source code, with specification of all dependencies, including external libraries. [Not Applicable]
  \end{enumerate}

  \item For any theoretical claim, check if you include:
  \begin{enumerate}
    \item Statements of the full set of assumptions of all theoretical results. [Yes]
    \item Complete proofs of all theoretical results. [Yes]
    \item Clear explanations of any assumptions. [Yes]     
  \end{enumerate}

  \item For all figures and tables that present empirical results, check if you include:
  \begin{enumerate}
    \item The code, data, and instructions needed to reproduce the main experimental results (either in the supplemental material or as a URL). [Yes]
    \item All the training details (e.g., data splits, hyperparameters, how they were chosen). [Not Applicable]
    \item A clear definition of the specific measure or statistics and error bars (e.g., with respect to the random seed after running experiments multiple times). [No]
    \item A description of the computing infrastructure used. (e.g., type of GPUs, internal cluster, or cloud provider). [No]
  \end{enumerate}

  \item If you are using existing assets (e.g., code, data, models) or curating/releasing new assets, check if you include:
  \begin{enumerate}
    \item Citations of the creator If your work uses existing assets. [Yes]
    \item The license information of the assets, if applicable. [No]
    \item New assets either in the supplemental material or as a URL, if applicable. [No]
    \item Information about consent from data providers/curators. [Not Applicable]
    \item Discussion of sensible content if applicable, e.g., personally identifiable information or offensive content. [Not Applicable]
  \end{enumerate}

  \item If you used crowdsourcing or conducted research with human subjects, check if you include:
  \begin{enumerate}
    \item The full text of instructions given to participants and screenshots. [Not Applicable]
    \item Descriptions of potential participant risks, with links to Institutional Review Board (IRB) approvals if applicable. [Not Applicable]
    \item The estimated hourly wage paid to participants and the total amount spent on participant compensation. [Not Applicable]
  \end{enumerate}

\end{enumerate}
\clearpage
\appendix
\thispagestyle{empty}
\onecolumn

\aistatstitle{Iterative Data Curation with Theoretical Guarantees: \\
Supplementary Materials}

\section{Additional Experiments and Figures} \label{appendix:figures}

\begin{figure}[h!]
    \centering
    \begin{subfigure}[b]{0.48\textwidth}
    \includegraphics[width=\linewidth]{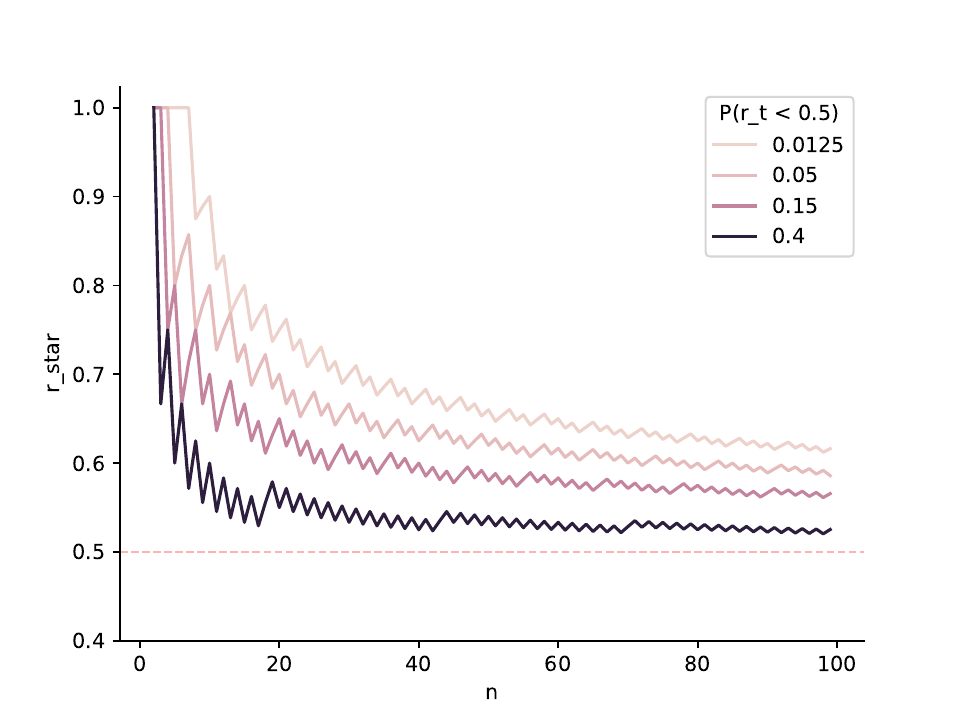}
    \end{subfigure}
    \hfill
    \begin{subfigure}[b]{0.48\textwidth}
    \includegraphics[width=\linewidth]{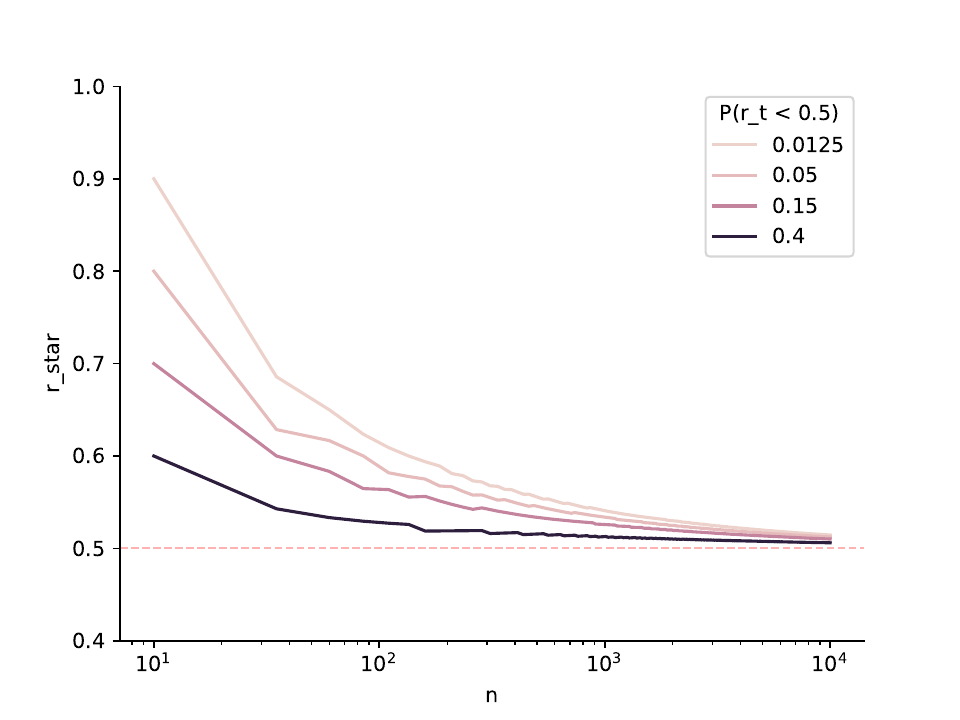}
    \end{subfigure}
    \caption{Convergence thresholds for different priors for $r$. The lines denote the smallest $m$ that guarantees convergence as a ratio of the sample size $m/n$. On the left, values $n \in \{2, 3, \dots, 99, 100\}$; on the right logarithmically spaced $n \in \{10^1, \dots, 10^4\}$. The red dashed line denotes the asymptote $m = n/2$.}
    \label{fig:graph-beta-1-1}
\end{figure}

\begin{figure}[h!]
    \centering
    \begin{subfigure}[b]{0.48\textwidth}
    \includegraphics[width=\linewidth]{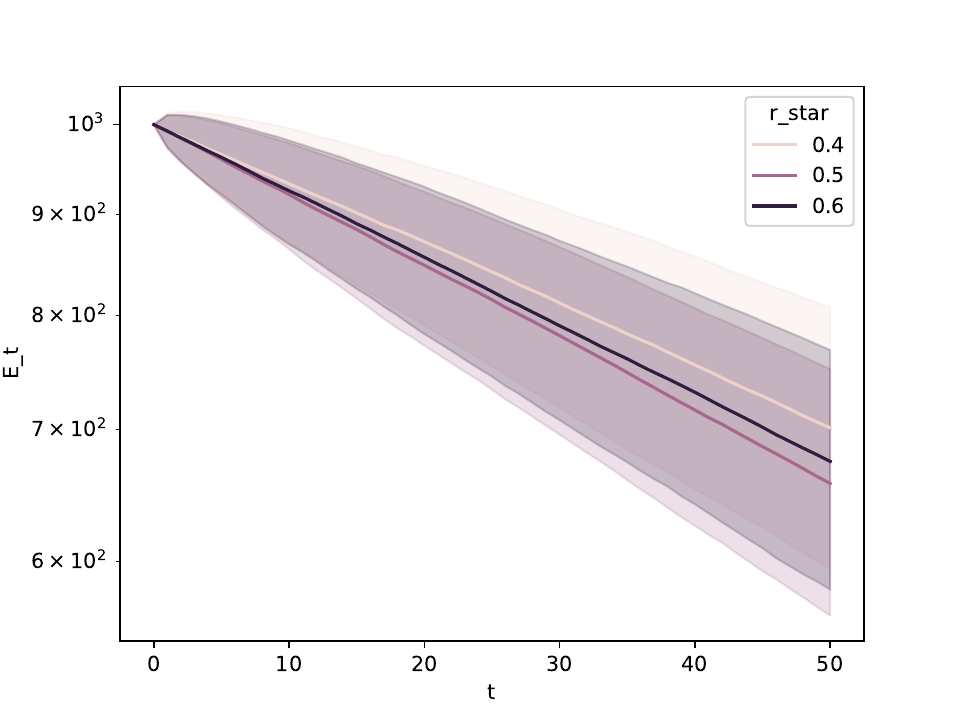}
    \caption{$r \sim Beta(2,1)$}
    \end{subfigure}
    \hfill
    \begin{subfigure}[b]{0.48\textwidth}
    \includegraphics[width=\linewidth]{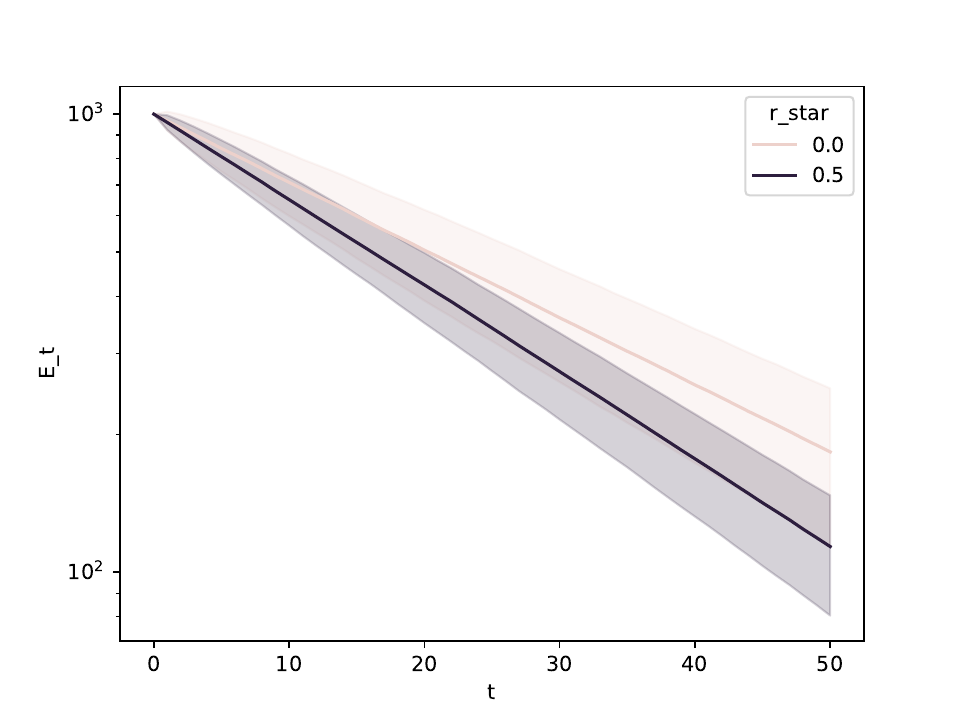}
    \caption{$r \sim Beta(1,2)$.}
    \end{subfigure}
    \caption{Number of errors in the data by iteration, $n=50$. $r_t \sim Beta(1,2)$. The y-axis is logarithmic to show the exponential decay of the errors.}
\end{figure}

\begin{figure}[h!]
    \centering
    \begin{subfigure}[b]{0.48\textwidth}
    \includegraphics[width=\linewidth]{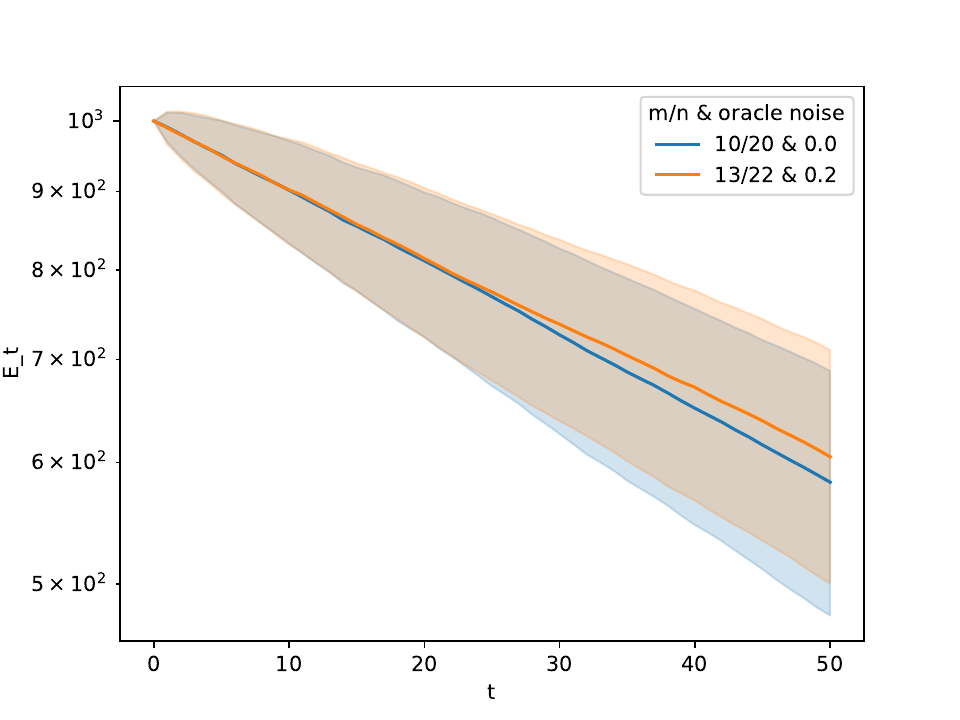}
    \caption{$r \sim Beta(1.5,0.8)$}
    \end{subfigure}
        \hfill
    \begin{subfigure}[b]{0.48\textwidth}
    \includegraphics[width=\linewidth]{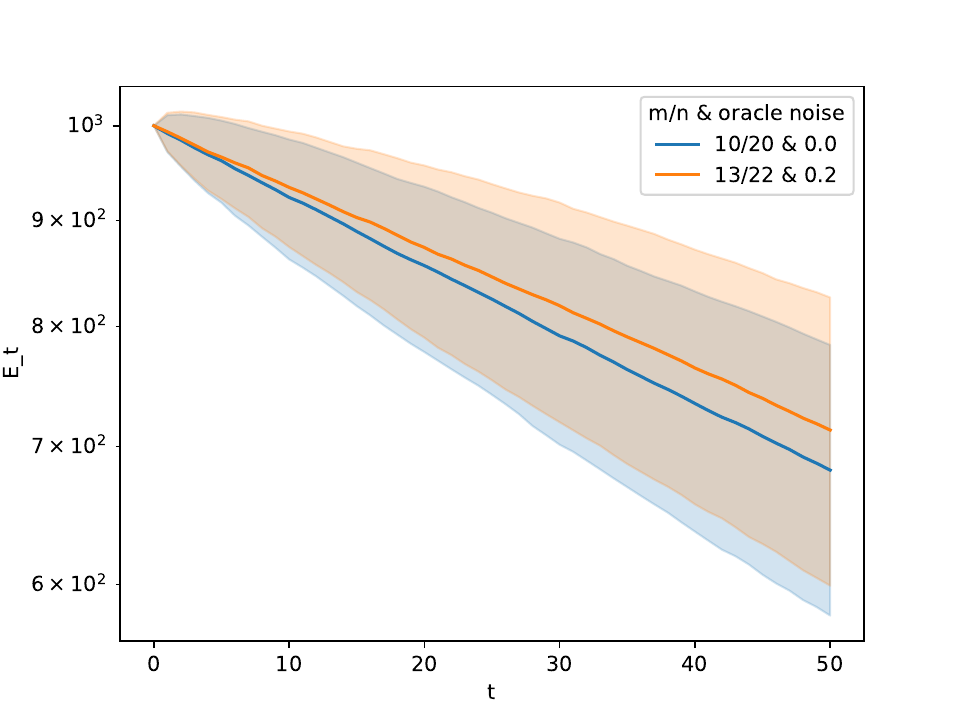}
    \caption{$r \sim Beta(1,2)$.}
    \end{subfigure}

    \caption{Number of errors in the data by iteration, $n=50$. Oracle and noisy oracle with noise rate $0.2$. $n$ and $m$ are set to correspond to the values of the theorem. The y-axis is logarithmic to show the exponential decay of the errors.}
    \label{fig:noisy_oracle}
\end{figure}

In Figure \ref{fig:noisy_oracle}, results from a simulation with and without a noisy oracle are presented. By theorem \ref{theorem:noisy_oracle}, $n$ and $m$ for the noisy oracle are set so that they correspond to the true oracle. Specifically, $n = 20 = n_{noisy} - \frac{\varepsilon}{1-\varepsilon} m_{noisy}$ and $m=10 = m_{noisy}/(1-\varepsilon)$. Since non-integer values are obtained but cannot be used , $n$ is rounded down, while $m$ is rounded up.

The figures show that $E_t$ under a noisy oracle  follows a rather similar distribution as $E_t$ under the true oracle.

\section{Proofs and detailed derivations} \label{appendix:detailed_proof}

\subsection{Derivations for Theorem \ref{theorem:zeroerrors}} 

The expectation of the number of descendants for one branch is
\begin{align}
    \mathbb{E}[\zeta \mid M = m]  &= \int_0^1 \int_0^1 \sum_{k=0}^2 k \cdot p_{k,t} d r_t d \lambda_t \\
    &=\int_0^1 \int_0^1 \left(1 - \lambda + 2(1-r_t)\lambda\right) p(r_t \mid M = m) d r_t d \lambda_t \\
    &=\int_0^1 \int_0^1 \left(1 - \lambda + 2 \lambda  -2 r_t \lambda\right) p(r_t \mid M = m) d r_t d \lambda_t\\
    &=\int_0^1 \int_0^1 \left(1 + \lambda  - 2 r_t \lambda\right) p(r_t \mid M = m) d r_t d \lambda_t\\
    &=1 + \int_0^1 \lambda -2 \lambda \int_0^1 r_t p(r_t \mid M = m) d r_t d \lambda_t\\
    &=1 + \mathbb{E} \Big[\lambda -2 \lambda \mathbb{E}[r_t \mid M = m]\Big] \\
    &=1 + 2 \mathbb E [\lambda] \left( 1/2  - \mathbb{E}[r_t \mid M = m] \right)
\end{align}
and subsequently
\begin{align}
    &\quad 1/2 - \mathbb{E}[r \mid M = m] \\
    &= \mathbb{E}[1/2 - r \mid M = m] \\
    &= \int_0^1 (1/2 - r) p(r \mid m) d r  \\
    & \propto \int_0^{0.5} (1/2 - r) p(m \mid r) p(r) d r + \int_{0.5}^1 (1/2 - r)p(m \mid r)  p(r) d r \\ 
    &\geq \int_0^{0.5} (1/2 - r) p(m \mid r) p(r) d r - \max_r {n \choose m} \left[(r - 1/2)(1-r)^mr^{n-m}\right] \int_{0.5}^1  p(r) d r  \\ 
    &\propto \int_0^{0.5} (1/2 - r) p(m \mid r) p(r) d r - \max_r {n \choose m} \left[(r - 1/2)(1-r)^mr^{n-m}\right] P(r \geq 0.5) \\
    &\geq a \int_0^{0.5} (1/2 - r) p(m \mid r) d r - \max_r {n \choose m} \left[(r - 1/2)(1-r)^mr^{n-m}\right] P(r \geq 0.5) \nonumber \\
    &\geq a \int_0^{1} (1/2 - r) p(m \mid r) d r - \max_r {n \choose m} \left[(r - 1/2)(1-r)^mr^{n-m}\right] P(r \geq 0.5) \\ 
    &= a /2 \int_0^{1} p(m \mid r) d r - a \int_0^{1} r p(m \mid r) d r \\
    &- \max_r {n \choose m} \left[(r - 1/2)(1-r)^mr^{n-m}\right] P(r \geq 0.5) \\
    &= a /2 \int_0^{1}{n \choose m} (1-r)^m r^{n-m} d r - a \int_0^{1} {n \choose m} (1-r)^m r^{n-m+1} d r \\
    &- \max_r {n \choose m} \left[(r - 1/2)(1-r)^mr^{n-m}\right] P(r \geq 0.5)\nonumber \\
    &= a {n \choose m} \left(1/2 B(n - m + 1, m + 1) - B(n - m + 2, m + 1) \right) \\
    &- \max_r {n \choose m} \left[(r - 1/2)(1-r)^mr^{n-m}\right] P(r \geq 0.5) \nonumber\\
    &= a {n \choose m} \left(1/2 \frac{2m - n}{n+2} B(n - m + 1, m+1)\right) - \max_r {n \choose m} \left[(r - 1/2)(1-r)^mr^{n-m}\right] P(r \geq 0.5)  \\ 
    &= a \left(1/2 \frac{2m - n}{n+2} \frac{1}{n+1}\right) - \max_r {n \choose m} \left[(r - 1/2)(1-r)^mr^{n-m}\right] P(r \geq 0.5)  \\
    &= a \left(\frac{2m - n}{2 (n+1) (n+2)}\right) - \max_r {n \choose m} \left[(r - 1/2)(1-r)^mr^{n-m}\right] P(r \geq 0.5) \\
    &\geq a \left(\frac{2m - n}{2 (n+1) (n+2)}\right) - \max_r {n \choose m}\left[(r - 1/2)(1-r)^mr^{n-m}\right] (1-  a/2)  \\
    &\propto a- \Bigg[\underbrace{\left(\frac{2 (n+1) (n+2)}{2m - n}\right) {n \choose m}\cdot   \quad\max_r \left[(r - 1/2)(1-r)^mr^{n-m}\right] }_{C(n, m)} \Bigg](1-  a/2) \\
    &= a - C(n, m) (1-  a/2) \\
    &\propto a(2 + C(n, m)) - C(n, m)\\
    &\propto a - C(n, m)/(2 + C(n, m))
\end{align}

\subsection{Derivations and properties of $C(n,m)$ for Theorem \ref{theorem:zeroerrors}}
\label{appendix:derivations}
To maximize ${n \choose m} (r-1/2)(1-r)^m r^{n-m}$ in the interval $(1/2, 1)$, we can maximize its logarithm since it's a product of positive factors
\begin{align}
    \log \Big[{n \choose m}(r-1/2)(1-r)^m r^{n-m} \Big] &=\log{n \choose m}+ \log(r-1/2) + m \log(1-r) + (n-m) \log r
\end{align}
This expression has the derivative
\begin{equation}
\begin{aligned}
    \frac{d}{dr}\log[{n \choose m}(r-1/2)(1-r)^m r^{n-m}] = \frac{1}{r -1/2} - \frac{m}{1-r} + \frac{n-m}{r}
\end{aligned}
\end{equation}
Its roots are the same when scaled with the expression $(r-1/2)(1-r)r$, which is positive in the range $(1/2, 1)$
\begin{align}
    &\frac{d}{dr}\log[(r-1/2)(1-r)^m r^{n-m}] = 0 \nonumber\\
    \iff &(r-1/2)(1-r)r\left(\frac{1}{r-1/2} - \frac{m}{1-r} + \frac{n-m}{r}\right) = 0\nonumber \\
    \iff & 2 m r - m + 2 n r^2 - 3 n r + n + 2 r^2 - 2 r = 0
\end{align}
This polynomial 
has the roots
\begin{align*}
    r = \frac{3n -2m + 2 \pm \sqrt{{(2m-n)^2 + 4(n + 1)}}}{4(n + 1)}
\end{align*}
If $n, m\geq1$ and $2n \geq n$, the larger root is larger than $1/2$ since
\begin{align}
     \frac{3n - 2m  +2 + \sqrt{(n-2m)^2 + 4(n + 1)}}{4(n + 1)} 
     &> \frac{3n - 2m  +2 + \sqrt{(n-2m)^2}}{4(n + 1)} \\
     &= \frac{3n - 2m  +2 - (n-2m)}{4(n + 1)} \quad \mid 2m \geq n \\
     &= \frac{2n  + 2}{4(n + 1)} = \frac{2(n +1)}{4(n + 1)} = \frac{1}{2} 
\end{align}
and smaller than $1$ since
\begin{align}
     \frac{3n - 2m  +2 + \sqrt{(n-2m)^2 + 4(n + 1)}}{4(n + 1)} &< \frac{3n - 2m  +2 + \sqrt{(n-2m)^2} +  2\sqrt{n+1}}{4(n + 1)} \\
     &= \frac{3n - 2m  +2 - (n-2m) + 2\sqrt{n+1}}{4(n + 1)} \\
     &= \frac{2n  + 2 + 2\sqrt{n+1}}{4(n + 1)} \\
     &= \frac{1}{2} + \frac{1}{2\sqrt{n+1}} \leq 1/2 + 1/2
\end{align}
since $\lVert x\rVert_1 \geq \lVert x\rVert_2$ and $2m \geq n$.

Similar logic can be used to show that the smaller root is always smaller than $1/2$. Thus, the derivative has only one root in the range $(1/2, 1)$, namely
\begin{equation}
       r = \frac{3n -2m + 2 + \sqrt{{(n-2m)^2 + 4(n + 1)}}}{4n + 4}
\end{equation}
As the expression $(r-1/2)(1-r)^{m}r^{n-m}$ is zero at the edges of the boundaries, and positive within the interval, its maximum is found at this root of the derivative. The maximum is
\begin{align}
\max_r [(r-1/2)(1-r)^{m}r^{n-m}] &= (\frac{3n -2m + 2 + \sqrt{{(n-2m)^2 + 4(n + 1)}}}{4n + 4} - 1/2) \nonumber\\
&\cdot \left(1 - \frac{3n -2m + 2 + \sqrt{{(n-2m)^2 + 4(n + 1)}}}{4n + 4} \right)^m \nonumber\\
&\cdot \left(\frac{3n -2m + 2 + \sqrt{{(n-2m)^2 + 4(n + 1)}}}{4n + 4}\right)^{n-m} \\
 = &\left(\frac{n -4m + 2 + \sqrt{{(n-2m)^2 + 4(n + 1)}}}{4n + 4} \right) \nonumber\\
\cdot &\left(\frac{n -2m + 2 - \sqrt{{(n-2m)^2 + 4(n + 1)}}}{4n + 4} \right)^m \nonumber\\
\cdot &\left(\frac{3n -2m + 2 + \sqrt{{(n-2m)^2 + 4(n + 1)}}}{4n + 4} \right)^{n-m} 
\end{align}
Further properties can be shown for the maximum $C(n, m)$. For instance, setting $m = n - d$ where $d \in \mathbb N$ and kept constant, $C(n, m)$ can be made arbitrarily small
\begin{equation}
    \begin{aligned}
        C(n, n - d) &= \frac{2(n+1)(n+2)}{2m -n} {n \choose n-d} \max_r [(r-1/2)(1-r)^{m}r^{n-m}] \\
        &< 2(n+1)(n+2) {n \choose n-d} \max_r [(r-1/2)(1-r)^{m}r^{n-m}] \\
        &< 2(n+2)^2 {n \choose n-d} \max_r [(r-1/2)(1-r)^{m}r^{n-m}] \\
        &< 2(n+2)^{2 + d} \max_r [(r-1/2)(1-r)^{m}r^{n-m}] \\
        &= 2(n+2)^{2 + d}\max_r [(r-1/2)(1-r)^{n-d}r^{d}] \\
        &\leq 2(n+2)^{2 + d}\max_r [(r-1/2)(1-r)^{n-d}] & \mid \text{ } r \leq 1 \\
        &\leq 2(n+2)^{2 + d} 1/2(1-1/2)^{n-d}  & \mid \text{ } 1/2 \leq r \leq 1 \\
        &\leq 2(n+2)^{2 + d} 1/2^{n-d+1} = O(2n^{2 + d} 2^{-n}) \to 0
    \end{aligned}
\end{equation}
i.e. $C(n,n-d)$ decays exponentially with growing $n$. Note that by the same steps, the same holds for even a ratio $m = q n$, where $q>0.5$, though with the exponent changing depending on the ratio
\begin{equation}
    \begin{aligned}
        C(n, q n) &= \frac{2(n+1)(n+2)}{2m -n} {n \choose qn} \max_r [(r-1/2)(1-r)^{m}r^{ n-m}] \\
        &< 2(n+2)^2 {n \choose qn} \max_r [(r-1/2)(1-r)^{m}r^{ n-m}] \\
        &\propto (n+2)^2 \left(\frac{1}{q^q (1-q)^{1-q}}\right)^n \max_r [(r-1/2)(1-r)^{m}r^{ n-m}] \\
        &< (n+2)^2 (q^\prime)^{-n} \max_r [(r-1/2)(1-r)^{m}r^{ n-m}]\\
        &= (n+2)^2 (q^\prime)^{-n} \max_r [(r-1/2)(1-r)^{qn}r^{(1-q)n}] \\
        &\leq (n+2)^2 (q^\prime)^{-n} 1/2^{qn} r_{max}^{(1-q)n} \\
        &\leq (n+2)^2 (q^\prime)^{-n} 1/2^{qn} {q^\prime}^{(1-q)n} \\
        &= (n+2)^2 (q^\prime)^{-qn} 1/2^{qn} \\
        &= O\left(n^2 (2 q^\prime)^{-n} \right) \to 0 \\
    \end{aligned}
\end{equation}
since $\forall q \quad \exists q^\prime < 1/2$ for which $\frac{1}{q^q (1-q)^{1-q}} < 1/q^\prime$.

\subsection{Derivations for Theorem \ref{theorem:tabulardata}} \label{appendix:tabular}

The additivity assumption in Theorem \ref{theorem:tabulardata}, which applies for $K$-dimensional vectors, can be directly shown
\begin{align*}
&d(S_t + u_1 + u_2, S^\star) - d(S_t, S^\star)\\
&= - \sum_{j=1}^K \mathbb I ( S_{t, j} + u_{1,j} + u_{2,j} = S^\star) - \mathbb I ( S_{t, j} = S^\star) \\
=& -\sum_{j=1}^K \mathbb I (u_{1,j} \neq 0 \land u_{2,j} = 0)\left( \mathbb I ( S_{t, j} + u_{1,j} = S^\star) - \mathbb I ( S_{t, j} = S^\star) \right)\\
&-\sum_{j=1}^K \mathbb I (u_{1,j} = 0 \land u_{2,j} \neq 0)\left( \mathbb I ( S_{t, j} + u_{2,j} = S^\star) - \mathbb I ( S_{t, j} = S^\star) \right)\\
&-\sum_{j=1}^K \underbrace{\mathbb I (u_{1,j} \neq 0 \land u_{2,j} \neq 0)}_{\text{always } 0}\left( \mathbb I ( S_{t, j}  u_{1,j} + u_{2,j} = S^\star) - \mathbb I ( S_{t, j} = S^\star) \right)\\
= &-\sum_{j=1}^K \mathbb I (u_{1,j} \neq 0 \land u_{2,j} = 0)\left( \mathbb I ( S_{t, j} + u_{1,j} = S^\star) - \mathbb I ( S_{t, j} = S^\star) \right)\\
&-\sum_{j=1}^K \mathbb I (u_{1,j} = 0 \land u_{2,j} \neq 0)\left( \mathbb I ( S_{t, j} + u_{2,j} = S^\star) - \mathbb I ( S_{t, j} = S^\star) \right)\\
=& d(S_t + u_1, S^\star) - d(S_t, S^\star) +d(S_t + u_2, S^\star) - d(S_t, S^\star)
\end{align*}
which holds for any data state $S_t$, true data $S^\star$, and pair of edits $(u_1, u_2)$.

\subsection{Properties of the longest common subsequence} \label{appendix:lcs}

We are unaware of this proof being found in the literature. However, as it is rather straightforward it is likely that we are not the first to prove it.

\begin{theorem}
If $S$ and $S^\prime$ are two sequences without duplicate elements, and for each two elements $(x,y)$ that are present in both $S$ and $S^\prime$, the relative order of such elements is the same in both sequences, then their longest common subsequence is unique.
\end{theorem}
    
\begin{proof}
Let $C$ be the set of elements common to both  $S$ and $S^\prime$. Let $\lvert C \rvert=k$. Any common subsequence of  $S$ and $S^\prime$ can only be formed using elements from $C$. Therefore, the length of the longest common subsequence (LCS) is at most $k$.

Consider the subsequence formed by taking all the elements of $C$ in the order they appear in $S$. Let this subsequence be $L_S$. Since for each pair of elements $(x, y) \in C \times C, x \neq y$, their relative order is the same in $S^\prime$ as it is in $S$, $L_S$ is also a subsequence of $S^\prime$. Thus, $L_S$ is a common subsequence of length $k$. Therefore, the length of the LCS is at least $k$.

Combining these two points, the length of the LCS is exactly $k$, the number of common elements.

Now, we need to show that this LCS is unique. Suppose there exists another common subsequence $L^\prime$ of length $k$ that is different from $L_S$. Since $L^\prime$ has length $k$, it must contain all the elements of $C$. However, because the relative order of the elements of $C$ is the same in both $S$ and $S^\prime$, any subsequence of length $k$ formed by elements of $C$ must have the elements in the same relative order. Therefore, $L^\prime$ must be the same as $L_S$.

Thus, the longest common subsequence is unique.
\end{proof}

\subsection{Stochastic dominance of the alternative posterior for the noisy oracle}

\begin{lemma} \label{lemma:alt_posterior}
The posterior proportional to $f(r \mid M=m) \propto  ((1-\varepsilon)(1-r) + \varepsilon)^m r^{n-m} p(r)$ is stochastically dominated by the posterior proportional to $g(r) = (1-r)^{(1 - \varepsilon)m} r^{n-m} p(r)$.
\end{lemma}

\begin{proof}
If a posterior proprtional to $f(x)$ has the monotone likelihood property wrt. another posterior proportional to $g(x)$
\begin{equation}
    \frac{d}{d x} \frac{f(x)}{g(x)} \geq 0 \quad \forall x
\end{equation}
the first posterior also stochastically dominates the second. 

We obtain the ratio
\begin{equation}
\begin{aligned}
f(x)/g(x) &= \frac{((1-\varepsilon) r + \varepsilon)^m (1-r)^{n-m} p(r)}{r^{ (1-\varepsilon)m} (1-r)^{n-m} p(r)} = \frac{((1-\varepsilon) r + \varepsilon)^m }{r^{ (1-\varepsilon)m} }
\end{aligned}
\end{equation}
We can investigate its derivative by taking the logarithm
\begin{align}
&\ln \frac{((1-\varepsilon) r + \varepsilon)^m }{r^{ (1-\varepsilon)m} } =m \ln ((1-\varepsilon)r + \varepsilon) - (1-\varepsilon) m \ln r
\end{align}
and further removing the multiplier $m \geq 1$
\begin{align}
\frac{d}{d r} \ln ((1-\varepsilon)r + \varepsilon) - (1-\varepsilon) \ln r &= \frac{1-\varepsilon}{(1-\varepsilon)r + \varepsilon} - \frac{1-\varepsilon}{r} \\
&\propto \frac{1}{(1-\varepsilon)r + \varepsilon} - \frac{1}{r} \\
&= \frac{1}{(1-\varepsilon)r + \varepsilon} - \frac{1}{(1-\varepsilon)r+\varepsilon r} 
\end{align}
Since $\varepsilon r \leq \varepsilon$ for all $r \in [0, 1]$, the derivative is nonnegative and the property holds.
\end{proof}


\section{Additional data types}



\subsection{Tabular data with an unknown number of indexed rows}

$\mathcal S$ are sets consisting of tuples $(e, k, v) \in \mathbb Z^L \times \{1, \dots, K\} \times \mathbb R$. Each of the elements maps the index of an element, and the index of the value to the value. For example, an element indexed $e=123$ at the column $k=7$ has the value $v=0.324$. This allows for an indexed tabular data format with an arbitrary number of elements with missing values. Furthermore, $\mathcal S$ is defined to not have duplicates in terms of $(e,k, \cdot)$.

A suitable distance metric is the size of the symmetric set difference
\begin{equation}
d(S, S^\prime)= \sum_{s \in S^\prime} \mathbb I(s \notin S) + \sum_{s \in S} \mathbb I(s \notin S^\prime)
\end{equation}
The symmetric set difference can also be used as $\Delta(S, S^\prime)$ from which to sample edits. Intuetively, each missing value is punished with 1, and each incorrect value with 1; mapping a valid index $(e,k,\cdot)$ to an incorrect value is thus punished with 2. This relies on the elements being unique.

Single additions have a straightforward effect on the distance
\begin{equation}
    \begin{aligned}
        d(S + u, S^\prime) - d(S, S^\prime) &= \sum_{s \in S^\star} \mathbb I(s \notin S \cup \{u\}) \\
        &+ \sum_{s \in S \cup \{u\}} \mathbb I(s \notin S^\star) \\
        &- \sum_{s \in S^\star} \mathbb I(s \notin S) - \sum_{s \in S} \mathbb I(s \notin S^\star) \\
        &= -\mathbb I(u \in S^\star) + \mathbb I(u \notin S^\star) \\
        &= 1 - 2 \mathbb I(u \in S^\prime)\\
    \end{aligned}
\end{equation}
Addition to $S$ corresponds to deletion from $S^\prime$, which yields a similar formula for the deletion of $u$
\begin{equation}
    d(S + u, S^\prime) - d(S, S^\prime) = 1 - 2 \mathbb I(u \notin S^\prime)\\
\end{equation}
Furthermore, the metric fulfils the additivity assumption
\begin{equation}
    \begin{aligned}
        &d(S + u_1 + u_2, S^\star) - d(S, S^\star) 
        \\&= \sum_{s \in S^\star} \mathbb I(s \notin S \cup \{u_1, u_2\}) + \sum_{s \in S \cup \{u_1, u_2\}} \mathbb I(s \notin S^\star) \\
        &- \sum_{s \in S^\star} \mathbb I(s \notin S) - \sum_{s \in S} \mathbb I(s \notin S^\star) \\
        &= \sum_{s \in S^\star} -\mathbb I(s \in S \cup \{u_1, u_2\}) + \sum_{s \in \{u_1, u_2\}} \mathbb I(s \notin S^\star) \\
        &= \sum_{s \in \{u_1, u_2\}} -\mathbb I(s \in S^\star) + \sum_{s \in \{u_1, u_2\}} \mathbb I(s \notin S^\star) \\
        &= 1 - 2 \mathbb I(u_1 \in S^\star) + 1 - 2\mathbb I(u_2 \in S^\star) \\
        &= d(S + u_1, S^\star) - d(S, S^\star) + d(S +  u_2, S^\star) - d(S, S^\star)
    \end{aligned}
\end{equation}

\section{Diff sampling of text} \label{appendix:diffsampling}

For the sampling of edits in text data, \texttt{git diff --stat} is run to obtain the number of changes by file. Then, a stratified sample is taken of these files to obtain the number of edits to sample per each file, which all sum up to $n$, resulting $n_f$ for all $f \in F$, which is the set of files sampled. Then, \texttt{git diff} is run on the files where is at least 1 (which is $n$ different files at maximum). The result of each of these commands is piped to a .diff file, and for each file $f$, $n_f$ diffs are uniformly sampled. The procedure can be fully automated and is provided in the Supplementary Material as a Python package.

An example of an individual sampled edit is provided below:

\begin{lstlisting}
data/1907/prot-1907--ak--023.xml

Diff starting from line [2155](https://github.com/swerik-project/riksdagen-records/tree/join-seg-II/data/1907/prot-1907--ak--023.xml#L2155)
\end{lstlisting}

\begin{lstlisting}[language=diff]
@@ -2239,10 +2155,7 @@
         </note>
         <note xml:id="i-GpA3qd9kHzpzc5cXwJybo1">
           Angaende en ny forbindelseled mellan Baggensfjarden och en inre
-          Stockholmsskargarden.
-        </note>
-        <note xml:id="i-P9vjoFjDP3bFydxB7oD7M8">
-          (Forts )
+          Stockholmsskargarden. (Forts )
         </note>
         <note xml:id="i-PyYHcEah3YgC4sBc2fUQiQ" type="date">
           26 Onsdagen den 13 Mars, e. m.
\end{lstlisting}

\begin{lstlisting}
- [x] Correct
- [ ] Incorrect
\end{lstlisting}

This particular edit was deemed correct, as it merged an accidentally split paragraph.

\end{document}